\documentclass[runningheads]{llncs}
\usepackage{graphicx}
\usepackage{shortcuts}
\usepackage{amsmath}
\usepackage{amsfonts}

\DeclareMathOperator{\argmin}{arg\,min}
\usepackage{algorithm}
\usepackage{algpseudocode}
\usepackage{subfig}
\usepackage{multirow}
\usepackage{cite}
\usepackage{hyperref}
\usepackage[title]{appendix}
\usepackage[misc]{ifsym}

\begin{document}
\title{\codename: Practical Privacy-Preserving Collaborative Machine Learning
}
%
%
\author{Yanjun Zhang \and
	Guangdong Bai \textsuperscript{(\Letter)}
	\and
	Xue Li \and
	Caitlin Curtis  \and\\
	Chen Chen\and
	Ryan K L Ko}
\authorrunning{Y. Zhang et al.}
%
\institute{The University of Queensland, St Lucia, Queensland, Australia \\
	\email{\{yanjun.zhang, g.bai, c.curtis, chen.chen, ryan.ko\}@uq.edu.au} \\
    \email{xueli@itee.uq.edu.au}}

%
\maketitle              
\begin{abstract}
Collaborative learning enables two or more participants, each with their own training dataset, to collaboratively learn a joint model.
It is desirable that the collaboration should not cause the disclosure of either the raw datasets of each individual owner or the local model parameters trained on them.
This privacy-preservation requirement has been approached through differential privacy mechanisms, homomorphic encryption~(HE) and secure multiparty computation~(MPC), but existing attempts may either introduce the loss of model accuracy or imply significant computational and/or communicational overhead.

In this work, we address this problem with the lightweight additive secret sharing technique. 
We propose \codename, a framework for protecting 
local data and local models while ensuring the correctness of training processes.
\codename employs secret sharing technique 
for securely evaluating addition operations in a multiparty computation environment, and achieves practicability by employing only the homomorphic addition operations.
We formally prove that it guarantees privacy preservation even though the majority ($n-2$ out of $n$) 
of participants are corrupted. 
With experiments on real-world datasets, we further demonstrate that \codename retains high efficiency.
It achieves a speedup of more than 45X over the state-of-the-art MPC-/HE-based 
schemes 
for training linear/logistic regression, and 216X faster for training neural network.
\end{abstract}

\section{Introduction}

The performance of machine learning largely relies on the availability of datasets.
To take advantage of massive data owned by multiple entities,
collaborative machine learning has been proposed to enable two or more data owners to construct a joint model.
One typical scenario demanding 
collaborative learning is
where the \emph{features} of a same sample are held by multiple data owners.
The collaboration among owners can improve the model accuracy by leveraging additional features from each other.
A real-world example is that a recommender system can take use of the ratings of a same item among multiple online merchants to enhance its predictive power.

To address the privacy concerns arising from collaborative learning, 
many studies~\cite{chen2018privacy, wang2019scalable, gascon2017privacy, zhang2019pefl} have been proposed 
to provide \emph{data locality} by distributing learning algorithms onto data owners such that the data can be confined within their owners.
Despite this, their learning processes still entail sharing 
locally trained models, in order to synthesize the final models. 
However these local models are subject to information leakage. For example, model-inversion attacks ~\cite{melis2019exploiting, fredrikson2015model} are able to restore training data from them. 
In addition, in the scenarios where the model itself represents intellectual property, e.g., in financial market systems, it is an essential requirement for  the local models to be kept confidential~\cite{papernot2016towards}.

To provide a supplementary, i.e., \emph{privacy-preserving synthesis of the local models}, differential privacy mechanisms and cryptographic mechanisms may be employed.
The former~\cite{zheng2019bdpl, hagestedt2019mbeacon, abadi2016deep, song2013stochastic, shokri2015privacy , dwork2014algorithmic, hu2019fdml} usually entails adding noise on the model parameters, 
causing loss in the accuracy of the final models.
The cryptographic mechanisms such as homomorphic encryption (HE)~\cite{chen2018privacy, marc2019privacy, sharma2019confidential} and secure multiparty computation (MPC)~\cite{ gascon2017privacy, gascon2016secure, liu2017oblivious, mohassel2017secureml, bonawitz2017practical}
are able to yield identical models as those trained on plaintext data, but are known to be limited by the significant computational or communicational overheads.

This work focuses on the practicability of the cryptographic solutions.
We propose a lightweight framework named \codename for privacy-preserving collaborative learning in the distributed feature scenario.
\codename adopts the two-layer architecture commonly used in previous privacy-preserving collaborative learning frameworks~\cite{mohassel2017secureml, hu2019fdml, wang2019scalable}.
It has a local node layer consisting of participating data owners, and an aggregation node~(which can be untrustworthy).
The main strategy of \codename is to dispense the homomorphic \emph{multiplication operations} and \emph{non-linear functions} on ciphertext, as they are far more costly in computation and communication than the \emph{addition operations}~\cite{gentry2009fully,albrecht2018homomorphic,yao1986generate, goldreich2019play}.
To this end, we redesign the workflow of collaborative learning, so that it employs only the homomorphic addition operations provided by \emph{additive secret sharing scheme}~\cite{bogdanov2008sharemind} for synthesizing local models and intermediate outcomes.
The computation that is carried out by the aggregation mode uses only the \emph{sum} of the intermediate results that are generated by the local nodes~(detailed in Section~\ref{sec:workflow}).
As such, \codename achieves significant cost savings, in comparison with state-of-the-art cryptographic solutions.

The redesigned workflow also ensures that both the raw data and local models are always kept with their owners.
We formally prove \codename preserves privacy in such a way that the \emph{honest-but-curious} participants, who have access to the sum of the intermediate outcomes produced by the local nodes and additional knowledge learned from the training iterations, are unlikely~(i.e., with a negligible probability $\varepsilon$) to reveal the raw 
training data or the local model parameters of other participants.

Notably, our new collaborative learning workflow in \codename introduces no sacrifice to the accuracy of the models, and also supports a wide range of machine learning algorithms as previous work does~\cite{abadi2016deep, mohassel2017secureml, hu2019fdml}. Intuitively, our solution makes use of the chain rules in calculus to decompose gradient descent optimization into computational primitives, and to distribute them to the local nodes and the aggregation node respectively.
When they collaborate together, these primitives can be recombined to achieve the correctness of learning.
We prove that such correctness is guaranteed for any algorithm that uses gradient descent for optimization, including but not limited to linear regression, logistic regression, and a variety of neural networks.

\paragraph{Contributions}
 In general, our contributions can be summarized as follows.
\begin{itemize}
	\item \textbf{A Novel Privacy-preserving Collaborative Framework.}
	We propose a novel framework \codename for collaborative learning with distributed features.
    It preserves privacy while enabling a wide range of machine learning algorithms and achieving high computation efficiency. Not only does \codename achieve the data locality as previous work does, but it also keeps the local models confidential.

	\item \textbf{Provable Privacy Preservation and Correctness Guarantee.}
	We prove the privacy preservation 
	of \codename, demonstrating a negligible probability of corrupted parties revealing either the original data or the trained parameters from other honest parties. 
	We also prove that \codename ensures the learned model is identical to that in the traditional non-distributed framework.

	\item \textbf{Experimental Evaluations.} We conduct  
	experiments on real datasets, showing that \codename achieves a significant improvement of efficiency over the state-of-the-art cryptographic solutions based on MPC and HE.
For example, \codename  achieves around $22.5$ minutes for a two-hidden-layer neural network to process all samples in the MNIST dataset~\cite{lecun-mnisthandwrittendigit-2010}, while it takes more than $81$ hours with a state-of-the-art MPC protocol SecureML~\cite{mohassel2017secureml}.
\end{itemize}


\section{Background}\label{sec:background}
In this section, we introduce the background knowledge that is necessary to understand our framework.

\subsection{Gradient Descent Optimization}
\label{sec:gd}
Gradient descent is by far the most commonly used optimization strategy among various machine learning and deep learning algorithms.
It is used to find the values of 
coefficients that minimize a cost function as far as possible.
Given a defined cost function $J$,  the coefficient matrix 
$W$ is derived by the optimization $\argmin_W J$, and is updated as: 
\begin{equation}
W := W - \alpha\frac{\partial{J}}{\partial{W}}
\end{equation}

Given a particular training dataset $X$ and a label matrix $y$, the cost function $J$ can be defined as $J(\sigma(XW), y, W)$, 
where $\sigma$ is determined by the learning model. For example, in logistic regression, $\sigma$ is usually a sigmoid function $1/(1+e^{-z})$, while in neural network, $\sigma$ is a composite function that is known as forward propagation. According to the chain rule in calculus, the gradient with respect to $W$ is computed as $\frac{\partial{J}}{\partial{\sigma}}\frac{\partial{\sigma}}{\partial{XW}}\frac{\partial{XW}}{\partial{W}} + \tau$, where $\tau$ is the gradient with respect to the regularization term, which is independent to $X$. Let 
$\Delta$ be $\frac{\partial{J}}{\partial{\sigma}}\frac{\partial{\sigma}}{\partial{XW}}$, and $\frac{\partial{XW}}{\partial{W}}$ equals to $X$, then the gradient with respect to $W$ can be written as:
\begin{equation}
\label{equ:chain}
\frac{\partial{J}}{\partial{W}} = \Delta X + \tau.
\end{equation}

As such, we can decompose the gradient descent optimization into $\Delta$ , $X$ and $\tau$, in which $\Delta$ is a function of $XW$.
This provides an algorithmic foundation for \codename's distribution of learning algorithms.


\subsection{Additive Secret Sharing Scheme}
Secret sharing schemes aim to securely distribute secret values amongst a group of participants. \codename employs the secret sharing scheme proposed by \cite{bogdanov2008sharemind}, which uses \emph{additive sharing} over $\mathbb{Z}_{2^{32}}$.
In this scheme, a secret value $srt$ is split to $s$ shares $E^1_{srt},...,E^s_{srt} \in \mathbb{Z}_{2^{32}}$  such that
\begin{equation}
	E^1_{srt}+E^2_{srt}+...+E^s_{srt} \equiv srt \mod 2^{32},
\end{equation}
and any $s-1$ elements $E^{i_1}_{srt},...,E^{i_{s-1}}_{srt}$ are uniformly distributed.
This prevents any participant who has part of the shares from deriving the value of $srt$, unless all participants join their shares.

In addition, the scheme has a homomorphic property that allows efficient and secure addition on a set of secret values $srt_1,...,srt_s$ held by corresponding participants $S_1,...,S_s$.
To do this, each participant $S_i$ executes a randomised sharing algorithm $Shr(srt_i, S)$ to split its secret $srt_i$ into shares $E^{1}_{srt_i},...,E^{s}_{srt_i}$, and distributes each $E^{j}_{srt_i}$ to the participant $S_j$.
Then, each $S_i$  locally adds the shares it holds, 
$E^{i}_{srt_1},...,E^{i}_{srt_s}$, 
to produce $\sum_{j=1}^sE^{i}_{srt_j}$ (denoted by $E^i$ for brevity). After that, a reconstruction algorithm $Rec(\{(E^{i}, S_i)\}_{S_i \in S})$, which takes $E^i$ from each participant and add them together, can be executed by an aggregator to reconstruct the $\sum_{i=1}^s srt_i$ without revealing any secret addends $srt_i$.


\section{Design of \codename}
\label{sec:design}

\subsection{Scope and Threat Model}
\label{sec:threat}
The involved parties in \codename are a set of local nodes (i.e., data owners) $S_1, ..., S_s$ and an aggregation node  $Agg$. 
Each local node holds part of features
of the training samples,
denoted by $X^l$ $(l \in \{1,...,s\})$, and the corresponding local model $W^l$ $(l \in \{1,...,s\})$ trained on $X^l$.
Each $X^l \in \mathbb{R}^{m \times d_l}$ is a $m \times d_l$ matrix representing $m$ training samples with $d_l$ features, and $W^l \in \mathbb{R}^{d_l \times k}$ is a
matrix of coefficients, where $k$ is the number of output classes.
In \codename, $m$ and $k$  are public and known by every party, and $d_l$  is private and only known by the corresponding data owner $S_l$.
We use $X$ to denote the vertical concatenation of the local training datasets $X^1, ..., X^s$.
Then we know that $X$ is a $m \times n$ matrix  
where $n = \sum_{l=1}^{s}d_l$ (i.e., the total number of features in $X$).
Since $d_l$  is private, $n$ is unknown 
unless all of the local nodes join their views.

\codename aims to 
defend against an honest-but-curious adversary $\mathcal{A}$, who follows the collaboration protocols and training procedures, but is intending to obtain the datasets of other local nodes~(i.e., $X^l$) and/or the model parameters trained out of them~(i.e., $W^l$).
The adversary may control $Agg$, and $t$ out of $s$ local nodes. Here, we conservatively assume $t < s-1$, which implies that at least two local nodes need to be out of the adversary's control, as $s-1$ comprised local nodes who has the sum of their shares, colluding with $Agg$ who has the sum of all shares, will be able to obtain the share of the remaining node by a simple subtraction
~(detailed in Section \ref{sec:security}).

\subsection{Definitions of Privacy Preservation and Correctness}
Keeping local data/model private and providing functional correctness are the main properties \codename aims to achieve. Below we present the definitions 
of these two properties.
\begin{definition}
	\label{def:security}
	\textbf{($\varepsilon$-privacy)}
A mechanism preserves $\varepsilon$-privacy if the probability for a probabilistic polynomial-time (PPT) adversary to derive $X^l$ or $W^l$ of any benign node $S_l$ 
based on its knowledge is not greater than $\varepsilon$. 
\end{definition}

\begin{definition}
	\label{def:correct}
	\textbf{(Correctness)}
	Given a function $\mathcal{F}$ that takes as input a training dataset $X$, and its distributed version $\mathcal{F^\prime}$ that takes as inputs $X$'s vertical partitions $X^1, ..., X^s$, we say
$\mathcal{F^\prime}$ is correct if $\mathcal{F^\prime}(\{(S_l, X^l), Agg\}_{S_l \in S})=\mathcal{F}(X)$.
\end{definition}

\subsection{Workflow of \codename}
\label{sec:workflow}

Figure \ref{fig:Gradient} illustrates the end-to-end workflow of \codename, 
which is divided into the following steps.
\begin{figure}
	\includegraphics[width=\textwidth]{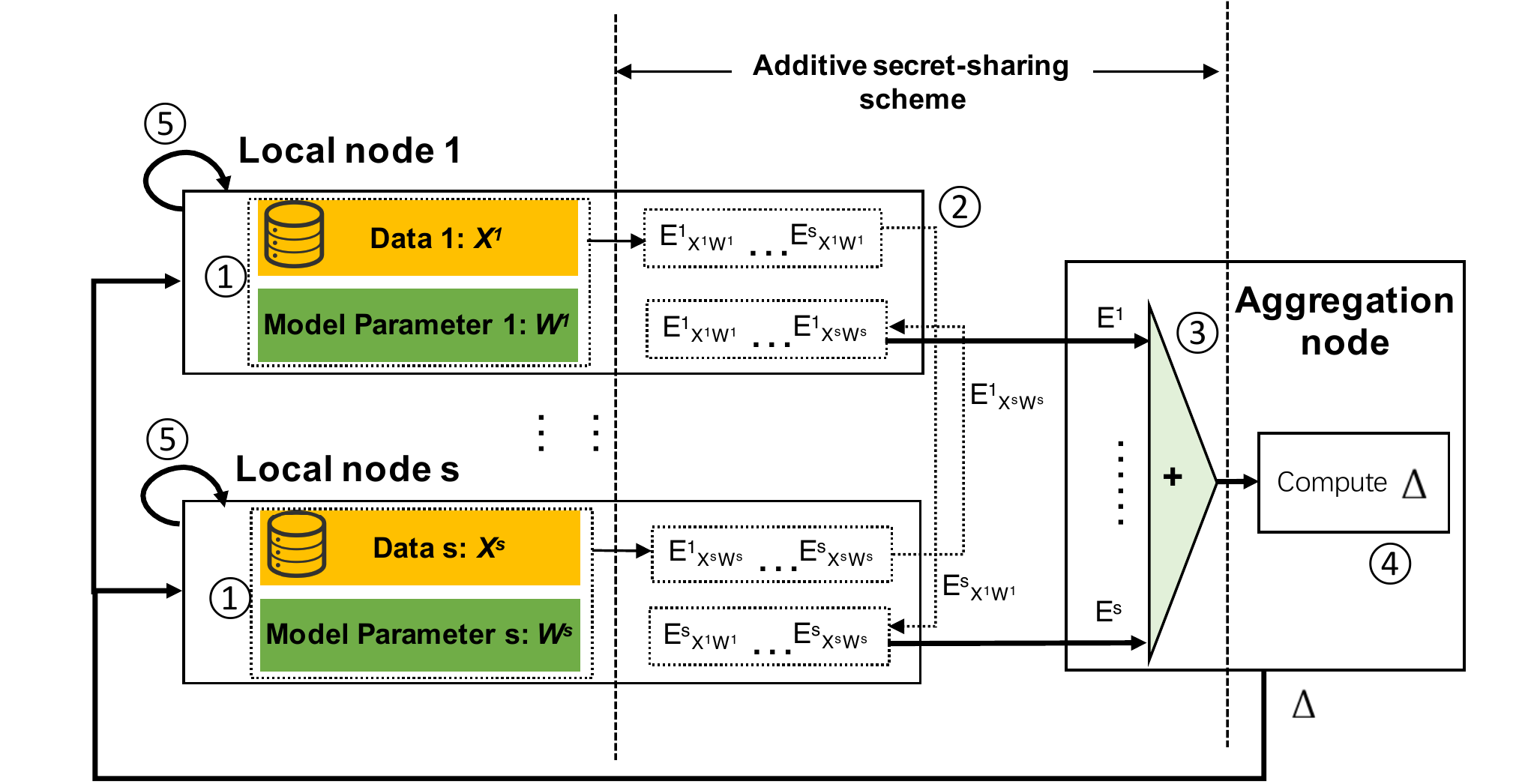}
	\centering
	\caption{The end-to-end workflow of the \codename} 
	\label{fig:Gradient}
\end{figure}

\textbf{Initialization}: Each local node $S_l$ holds its own training dataset $X^l \in \mathbb{R}^{m \times d_l}$, and randomly initializes its coefficient matrix 
$W^l \in \mathbb{R}^{d_l \times k}$.

\textbf{Step 1}: In each iteration of gradient descent, each $S_l$ multiplies $X^l$ by $W^l$ locally,   resulting in a $X^lW^l \in \mathbb{R}^{m \times k}$.
The value of $d_l$, i.e., the number of features in $X^l$, is removed by such a matrix multiplication.

\textbf{Step 2}: Each $S_l$ executes the sharing algorithm $Shr$ to split $X^lW^l$ into $s$ shares using the \emph{additive secret sharing scheme}.
\begin{equation}
{\{(S_i, E^i_{X^lW^l})\}_{S_i \in S}} \leftarrow Shr(X^lW^l, S),
\end{equation}
in which $Shr$  takes as input a secret $X^lW^l$ and a set
$S$ of local nodes, and produces a set of shares $E^i_{X^lW^l} (i \in \{1,...,s\})$, each of which is distributed to a different $S_i \in S$.  Then, each  $S_l$ calculates the sum of all shares it receives, and gets $E^l = \sum_{j = 1}^s E_{X^jW^j}^l$.
\textbf{Step 3}:
$Agg$ collects all $E^l$ from local nodes and add them together. Since
\begin{equation}
\sum_{l = 1}^s E^l = \sum_{l = 1}^s \sum_{j = 1}^s E_{X^jW^j}^l = \sum_{j = 1}^s \sum_{l = 1}^s E_{X^jW^j}^l = \sum_{j = 1}^s X^jW^j,
\end{equation}
this addition reconstructs the homomorphic addition result $XW$ (equals $\sum\limits_{l=1}^{s} X^lW^l$).

\textbf{Step 4}: 
$Agg$ computes $\Delta = \frac{\partial{J}}{\partial{\sigma}}\frac{\partial{\sigma}}{\partial{XW}}$ (\emph{c.f.}, Equation \ref{equ:chain}), 
and sends $\Delta$ back to the local nodes.

\textbf{Step 5}: With the received $\Delta$, each $S_l$ updates its local coefficient matrix 
$W^l$.

\begin{equation}
W^l \leftarrow W^l - \alpha(\Delta X^l + \tau^l)
\end{equation}

\textbf{Step 1-5} are repeated for the next training iteration until convergence.

\section{Privacy Preservation Analysis}
\label{sec:security}

In this section, we analyze the privacy preservation of \codename's learning process.
To this end, we first investigate the overall knowledge that can be learned by an adversary from the training iterations (Section \ref{sec:party}).
Then we prove that the knowledge is limited such that the desired $\varepsilon$-privacy~(Definition \ref{def:security}) is achieved with a negligible $\varepsilon$ (Section \ref{sec:maintheorem}).


\subsection{Preliminaries}
\label{sec:pre}
We start with the following lemma 
that is soon used in our proof.

\begin{lemma}
	\label{lemma:eigen}
	Consider a positive (semi-)definite matrix $A$ that is obtained as the product of a real number matrix $B$ by its transpose $B^T$
	\begin{equation}
	A = BB^T,
	\end{equation}
	where $B$ is of rank $r$.  
	Without knowing the number of columns of $B$, the probability of solving the $B$ given $A$, denoted as $P(B|A)$, is $ \leq (r!)^{-1}$.  
\end{lemma}

\begin{proof}
	Since the matrix $A$ is positive (semi-)definite, 
	there exists an eigen-decomposition such that 
	\begin{equation}
	\label{equ:eigen}
	A = U\Lambda U^T,
	\end{equation}
	where $U$ denotes a matrix 
	of eigenvectors of $A$ (each column of $U$ is an eigenvector of $A$), and $\Lambda$ denotes a diagonal matrix whose diagonal elements are the eigenvalues.
For a positive (semi-)definite matrix, all the eigenvalues are non-negative.
A different ordering of the eigenvector columns results in a different $U$ and a corresponding $\Lambda$~\cite{golub1996matrix}.
	
	With the eigen-decomposition, $B$ can be constructed by $B = U\Lambda^{\frac{1}{2}}$,
	where $\Lambda^{\frac{1}{2}}$ has the square roots of eigenvalues as its diagonal elements, and all its remaining values are zeros.
Each eigen-decomposition leads to a different $U$ and a corresponding $\Lambda$, resulting in a unique solution of $B$.
A matrix $B$ of rank $r$ has $r$ non-zero eigenvalues and thus there are $r!$ different possible orderings of eigenvector columns,
implying $r!$ different $U$s and the corresponding $\Lambda$s.
Consequently, there are $r!$ different possible solutions of computing $B$, which gives the probability of solving $B$ given $A$ with eigen-decomposition $P_{eigen}(B|A) = (r!)^{-1}$. 
	
	In addition, without the knowledge of the number of columns in $B$, there are more than $r!$ possible solutions of solving $B$.
	This is because from the eigen-decomposition construction, $B$'s columns are orthogonal. 
	In general, the matrix $B$ need not have orthogonal columns (it can be rectangular)~\cite{horn2012matrix}.
	Thus, $P(B|A) \leq (r!)^{-1}$. 
\end{proof}

\subsection{Party Knowledge}
\label{sec:party}

We define the \emph{party knowledge} as the overall knowledge that can be learned by adversary parties. It includes the parties' own inputs, and the additional knowledge that can be inferred from the training iterations.
We prove that the {party knowledge} in \codename is bounded within a certain range.
In particular, we demonstrate that the overall \emph{party knowledge} of adversary party $I$ in \codename is a set of $\{X^{l^\prime}\mid_{l^\prime \in I}, W^{l^\prime}\mid_{l^\prime \in I}, XW, \sum X^lW^l \mid_{l \in S \setminus I}, \sum X^l(X^l)^T \mid_{l \in S \setminus I}, XX^T\}$, where $\{X^{l^\prime}\mid_{l^\prime \in I}, W^{l^\prime}\mid_{l^\prime \in I}, XW\}$ are the adversary's own input  in the workflow, and $\{\sum X^lW^l \mid_{l \in S \setminus I}, \sum X^l(X^l)^T \mid_{l \in S \setminus I}, XX^T\}$ are the additional information that can be inferred from the training iterations.
The party knowledge we derive in this Section will be used in Section \ref{sec:maintheorem} to prove the privacy preservation of \codename.

We use the simulation paradigm (also known as the \textit{real/ideal model})~\cite{goldreich2019play} to prove such a bound of party knowledge.
The simulation paradigm compares what an adversary can do in a real protocol execution $REAL$ to what it can do in an ideal setting with a trusted functionality (simulation) $SIM$~\cite{goldreich2019play}.
Formally, the protocol $\mathcal{P}$ securely computes a functionality  
$\mathcal{F}_{\mathcal{P}}$ if for every adversary in $REAL$, there exists an adversary in $SIM$, such that the view of the adversary from  $REAL$ is indistinguishable from the view of the adversary from $SIM$.
A perfect indistinguishability~\cite{canetti2008theory} between the view of $REAL$ and $SIM$ guarantees that the adversary, without error probability, can learn nothing more than their own inputs and the information required by $SIM$ for the simulation.

We introduce some notations used in our proof.
We use $X^{S^\prime} = \{X^l|S_l \in S^\prime\}$ 
to indicate the inputs of any subset of local nodes $S^\prime \subseteq S$.
Given any subset $I_0 \subseteq S$ of the parties \emph{without} the knowledge of $XW$, and subset $I_1 \subseteq S \cup \{Agg\}$ of the parties \emph{with} the knowledge of $XW$,
let $REAL (X^S, XW, S, t, \mathcal{P}, I)$ denote the combined views of all parties in $I = I_0\cup I_1$ from the execution of a real protocol $\mathcal{P}$, where $t$ is the adversary threshold (recall that $t < s-1$). Let $SIM(X^{I}, Z, S, t, \mathcal{F}_{\mathcal{P}}, I )$ denote the views of $I$ from an ideal execution, where $Z$ is the information required by $SIM$ for the simulation.
In other words, the $Z$ indicates the party knowledge that the adversary can and only can learn other than their own inputs.
\begin{theorem}
	\label{theo:party}
	\textbf{(Party Knowledge)}
The simulator $SIM(X^I, Z, S, t, \mathcal{F}_{\mathcal{P}}, I)$ is perfectly indistinguishable from $REAL$ with respect to their outputs, namely
	\[ REAL (X^S, XW, S, t, \mathcal{P}, I) \equiv SIM(X^I, Z, S, t, \mathcal{F}_{\mathcal{P}}, I)\]
	if and only if
	\[Z = \{z_1 = \sum X^lW^l\mid_{l\in S \setminus  I},  z_2 = \sum X^l(X^l)^T\mid_{l\in S \setminus  I}, z_3 = XX^T\}.\] 
\end{theorem}

\begin{proof}
	We define the simulator through each of the $i^{th}$ training iteration as:
	
	\begin{itemize}
		\item $SIM_0$: This is the simulator for the \textbf{Initialization}. 
		In the step of initialization, the view of parties in $ I_0 \cup I_1$ does not depend on the inputs of the parties not in $I_0 \cup I_1$. Therefore, instead of sending the actual $\{X^lW^{l(0)}\}_{l \in S \setminus  \{ I_0 \cup I_1\}}$ of the parties $S \setminus  \{ I_0 \cup I_1\}$ to the aggregation node, the simulator can produce a simulation by running the parties $S \setminus \{ I_0 \cup I_1\}$ on a pseudorandom vector $\mu^{l(0)}$ in $\mathbb{R}^{m}$ as input, and then output the same pseudorandom vector to the aggregation node. 
		Since the model parameter $W^l$ is also randomized in the step of initialization, 
		the pseudorandom vectors for the inputs of all honest parties $S \setminus  \{ I_0 \cup I_1\}$, and the joint view of parties in $\{ I_0 \cup I_1\}$ will be identical to that in $REAL$
		\[ \{\mu^{l(0)}\mid_{l\in S \setminus \{ I_0 \cup I_1\}}, X^{l^\prime}W^{l^\prime (0)}\mid_{l^\prime\in   \{ I_0 \cup I_1\}}\} \equiv \{X^SW^{S(0)}\}.\]
		\item $SIM_i, i\geq 1$:
		This is the simulator for the $i^{th}$ training iteration ($i\geq 1$).
		The simulator computes  $\Delta^{(i)} = f(\Sigma^{(i)})$, where the function $f$ is determined by the learning model. For example, in the linear regression, $f(x) = x$, while in logistic regression, $f(x) = 1/(1+e^{-x})$.
		
		We respectively consider the simulator for $I_0$, $I_1$.		
		First, with respect to 
		$I_0$, the simulator computes $\Sigma^{(i)}$ as
		\[ \Sigma^{(i)} = \sum \mu^{l(i)}\mid_{l\in S \setminus  I_0} + \sum X^{l^\prime}W^{l^\prime (i)}\mid_{l^\prime\in   I_0} - y,\] 
		
		where $\mu^{l(i)}$ is computed by the result from the previous iteration as
		\[ \mu^{l(i)} = z_1^{l(i-1)} - \frac{\alpha}{m}z_2^l\Delta^{(i-1)}\mid_{l\in S \setminus  I_0}.\]
		Therefore,
		\[ \sum \mu^{l(i)}\mid_{l\in S \setminus  I_0} = \sum z_1^{l(i-1)} - \frac{\alpha}{m} \sum z_2^l\Delta^{(i-1)}\mid_{l\in S \setminus  I_0}.\]
		Note $X^{l^\prime}W^{l^\prime (i)}$ is also computed by the result from the previous iteration, and
		\[ \sum X^{l^\prime}W^{l^\prime (i)}\mid_{l^\prime\in   I_0} = \sum X^{l^\prime}W^{l^\prime (i-1)} - \frac{\alpha}{m} \sum X^{l^\prime}(X^{l^\prime})^T\Delta^{(i-1)}\mid_{l^\prime\in   I_0}.\]

		Then, the $\Sigma^{(i)}$ can be written as
		\begin{align*}
		\Sigma^{(i)} 
		= &\sum z_1^{l(i-1)}\mid_{l\in S \setminus  I_0} + \sum X^{l^\prime}W^{l^\prime (i-1)}\mid_{l^\prime\in   I_0}\\
		&- \frac{\alpha}{m}(\sum z_2^l\Delta^{(i-1)}\mid_{l\in S \setminus  I_0} + \sum X^{l^\prime}(X^{l^\prime})^T\Delta^{(i-1)}\mid_{l^\prime\in   I_0})  - y.
		\end{align*}

		
		Note that, in $REAL$, the output of $\Sigma^{(i)}$ by $S$ is
		\[ \sum X^SW^{S(i)} - y = \sum X^SW^{S(i-1)} - \frac{\alpha}{m} \sum X^s(X^s)^T\Delta^{(i-1)} - y.\]
		
		Thus, \{$\sum z_1^{l(i-1)}\mid_{l\in S \setminus  I_0}, \sum z_2^l\mid_{l\in S \setminus  I_0}, \sum X^{l^\prime}W^{l^\prime (i-1)}\mid_{l^\prime\in   I_0}, \sum X^{l^\prime}(X^{l^\prime})^T\mid_{l^\prime\in   I_0}$\} which is the joint view of  all honest parties $S \setminus  I_0$ and parties in $I_0$ will be  perfectly indistinguishable 
		to $\{\sum X^SW^{S(i-1)}, \sum X^s(X^s)^T\}$ which is the output in $REAL$.
		
		Next, we consider the simulator for $I_1$. With respect to $I_1$,  we let the simulator compute $\Sigma^{(i)}$ as
		\[ \Sigma^{(i)} = \mu^{(i)} - y,\]
		where $\mu^{(i)}$ is computed by the result from the previous iteration as
		\[  \mu^{(i)} = XW^{(i-1)} - \frac{\alpha}{m}z_3\Delta^{(i-1)}.\]
		
		Then, the $\Sigma^{(i)}$ can be written as
		\[ \Sigma^{(i)} = XW^{(i-1)} - \frac{\alpha}{m}z_3\Delta^{(i-1)} - y,\]
		
		Note that, in $REAL$, the output of $\Sigma^{(i)}$ by $S$ is
		\[  XW^{(i)} - y =  XW^{(i-1)} - \frac{\alpha}{m}  XX^T\Delta^{(i-1)} - y.\]
		
		Thus, the joint view of all parties in $SIM$ with knowledge of $z_3$ will be perfectly indistinguishable to $\{ XW^{(i-1)},  XX^T\}$ which is the output in $REAL$.
		
	\end{itemize}
	
	All in all, the output of the simulator $SIM$ of each training iteration is perfectly indistinguishable from the output of $REAL$.
	For the simulator with respect to $I = I_0 \cup I_1$, 
knowledge of $z_1 = \sum z_1^{l(i-1)}\mid_{l\in S \setminus  I} = \sum X^lW^l\mid_{l\in S \setminus  I} $, $z_2 = \sum z_2^l\mid_{l\in S \setminus  I} = \sum X^l(X^l)^T\mid_{l\in S \setminus  I}$ and $z_3 = XX^T$ is sufficient, completing the proof.
	
\end{proof}


\subsection{Privacy Preservation Guarantee}
\label{sec:maintheorem}
With lemma introduced in Section \ref{sec:pre}, and party knowledge discussed in Section \ref{sec:party}, we give our theorem of privacy preservation guarantee.
\begin{theorem}
	\label{theo:security}
	Let $I$ denote the adversary party with party knowledge $\{X^{l^\prime}\mid_{l^\prime \in I,} W^{l^\prime}\mid_{l^\prime \in I}, XW, \sum X^lW^l \mid_{l \in S \setminus I}, \sum X^l(X^l)^T \mid_{l \in S \setminus I}, XX^T\}$. Let $X^H$ denote the vertical concatenation of $\{X^{l}\mid_{l \in S \setminus I}\}$, which is the concatenation of honest local nodes' training datasets. Let $r$ denote the rank of  $X^H$.
	\codename preserves $\varepsilon$-privacy against adversary party $I$\ 
	on the training dataset $X \in \mathbb{R}^{m \times n}$, 
	where $\varepsilon \leq (r!)^{-1}$.
\end{theorem}
\begin{proof}
	We give some sketches here. 
	
	We start with the party knowledge $XX^T$.  In \codename, $n$ (i.e. the number of column of $X$) is unknown given the adversary threshold $t < s-1$. In addition, the rank of $X$, denoted as $R$, is $> r$. Invoking Lemma \ref{lemma:eigen} gives that the probability of solving the $X$ of rank $R$ given $XX^T$ is 	$\leq (R!)^{-1}<(r!)^{-1}$. Then we combine the party knowledge $XW$ ($X \in \mathbb{R}^{m \times n}$, and $W \in \mathbb{R}^{n \times k}$). From Theorem \ref{theo:party}, $XX^T$ is the only information required by $SIM$ with respect to $I_1$ with the knowledge of $XW$. In other words, combining the $XW$ gives no more information other than $XX^T$. Therefore, given an unknown $n$, and the probability of solving the $X$ $< (r!)^{-1}$, we have the probability of solving $W$ is also $< (r!)^{-1}$.

	Then we continue to combine the party knowledge of $\{\sum X^lW^l \mid_{l \in S \setminus I}, \\ \sum X^l(X^l)^T \mid_{l \in S \setminus I}\}$. They are the sum of real number matrices and the sum of non-negative real number matrices respectively. Given the adversary threshold $t < s-1$, which means the number of honest local nodes (i.e. $|l|_{l \in S \setminus  I}$) is $\geq 2$, we have a negligible probability to derive any $X^lW^l$ or $X^l(X^l)^T$ from their sum. In addition, $d_l$ (the number of columns of $X^l$) is also unknown to $I$. Therefore, with the probability of solving the $X$ $< (r!)^{-1}$, the probability of solving either $X^l$ or $W^l$ is also $< (r!)^{-1}$.
	
	At last, we combine the party knowledge of $\{X^{l^\prime}\mid_{l^\prime \in I}, W^{l^\prime}\mid_{l^\prime \in I}\}$. First, they are the input of $I$, which are independent of \{$\sum X^lW^l \mid_{l \in S \setminus I}, \sum X^l(X^l)^T \mid_{l \in S \setminus I}$\}. Next, we combine them into $X(X)^T$, which will give the adversary the problem of solving $X^H(X^H)^T$. As each of $d_l \mid_{l \in S \setminus I}$ is unknown to $I$, we have the number of column of $X^H$ is also unknown to $I$. Thus, with Lemma \ref{lemma:eigen}, we have the probability of solving the $X^H$ of rank $r$  is $\leq (r!)^{-1}$. Similarly, combing $\{X^{l^\prime}\mid_{l^\prime \in I}, W^{l^\prime}\mid_{l^\prime \in I}\}$ to $XW$ will also give the probability of solving $W^H$ $\leq (r!)^{-1}$. 

	Thus, \codename preserves $\varepsilon$-privacy against adversary parties $I$ on $X$, and $\varepsilon \leq (r!)^{-1}$.
\end{proof}

With Theorem \ref{theo:security}, we demonstrate that, with a sufficient rank of the training dataset of honest parties, e.g. $\geq$ 35, which is common in real-world datasets, \codename achieves $10^{-40}$-privacy.

\section{Correctness Analysis and Case Study}
\label{sec:corr&case}
In this section, we first prove the correctness of \codename
when distributing learning algorithms that are based on gradient descent optimization.
Then we use a recurrent neural network as a case study to illustrate the collaborative learning process in \codename.
\subsection{Correctness of \codename's Gradient Descent Optimization}
\label{sec:correct}
The following theorem demonstrates that if a non-distributed gradient descent optimization algorithm taking $X$ as input, denoted by $\mathcal{F}_{GD}(X)$, converges to a local/global minima $\eta$, then executing \codename with the same hyper settings (such as cost function, step size, and model structure) on $X^1, ..., X^s$, denoted by $\mathcal{F}^{\prime}_{GD}(\{(S_l, X^l), Agg\}_{S_l \in S})$, also converges to $\eta$.
\begin{theorem}
	\label{theo:correct}
	\codename's distributed algorithm of solving gradient descent optimization  $\mathcal{F}^{\prime}_{GD}(\{(S_l, X^l), Agg\}_{S_l \in S})$ is correct. 	
\end{theorem}

\begin{proof}
	Let $\mathcal{F}_{GD}(X) = \eta$ denote the convergence of $\mathcal{F}_{GD}(X)$ to the  local/global minima $\eta$. Let $W_i$ denote the model parameters of $\mathcal{F}_{GD}$  at $i^{th}$ training iteration. Let $W_i^\prime = |\{W^l_i \}|_{l \in \{1,..,s\}}$ denote the vertical concatenation on $\{W^l_i \}_{l \in \{1,..,s\}}$, i.e., the model parameters of $\mathcal{F}^\prime_{GD}$  at $i^{th}$ training iteration.
		
	
	
	In $\mathcal{F}_{GD}$, the $i^{th}$ $(i \geq 1)$ training iteration update $W_i$ such that
	\begin{equation}
	\label{equ:correct1}
	\begin{aligned}
	W_i & =  W_{i-1} - \alpha\frac{\partial{J}}{\partial{W}}  =  W_{i-1} - \alpha\frac{\partial{J}}{\partial{(XW_{i-1})}} \frac{\partial{XW_{i-1}}}{\partial{W_{i-1}}}\\
	& =  W_{i-1} - \alpha\frac{\partial{J}}{\partial{(XW_{i-1})}} X.
	\end{aligned}
	\end{equation}
In $\mathcal{F}^{\prime}_{GD}$, each node $S_l$ updates its local $W^l_i$ in
	\begin{align*}
	W^l_i & =  W^l_{i-1} - \alpha \Delta X^l  =  W^l_{i-1} - \alpha\frac{\partial{J}}{\partial{(XW^\prime_{i-1})}} \frac{\partial{X^lW^l_{i-1}}}{\partial{W^l_{i-1}}}\\
	& =  W^l_{i-1} - \alpha\frac{\partial{J}}{\partial{(XW^\prime_{i-1})}} X^l.
	\end{align*}
	Since $W_i^\prime = |\{W^l_i \}|_{l \in \{1,..,s\}}$, and $X = |\{X^l \}|_{l \in \{1,..,s\}}$, we have in $\mathcal{F}^{\prime}_{GD}$ that
	\begin{equation}
	\label{equ:correct2}
	\begin{aligned}
	W_i^\prime & = |\{(W^l_{i-1} - \alpha\frac{\partial{J}}{\partial{(XW^\prime_{i-1})}} X^l) \}|_{l \in \{1,..,s\}}
	=  W_{i-1}^\prime- \alpha\frac{\partial{J}}{\partial{(XW^\prime_{i-1})}} X.
	\end{aligned}
	\end{equation}
	Comparing Equation \ref{equ:correct1} and \ref{equ:correct2},
	we can find that $W_i$ and $W_i^\prime$ are updated using the same equation.
Therefore, with $\mathcal{F}_{GD}(X) = \eta$,  the gradient descent guarantees $\mathcal{F}^{\prime}_{GD}(\{(S_l, X^l), Agg\}_{S_l \in S})$ also converges to $\eta$.
	
	Thus, $\mathcal{F^\prime}_{GD}(\{(S_l, X^l), Agg\}_{S_l \in S}) = \mathcal{F}_{GD}(X)$.
\end{proof}

\subsection{Case Study}
\label{sec:case}

Figure \ref{fig:rnn} shows an example of a two-layer feed-forward recurrent neural network (RNN). Every neural layer is attached with a time subscript $c$. 
The weight matrix $W$ maps the input vector $X^{(c)}$ to the hidden layer $h^{(c)}$. The weight matrix $V$ propagates the hidden layer to the output layer $\hat{y}^{(c)}$. The weight matrix $U$ maps the previous hidden layer to the current one. 

\begin{figure}
	\centering
	\includegraphics[width=10cm]{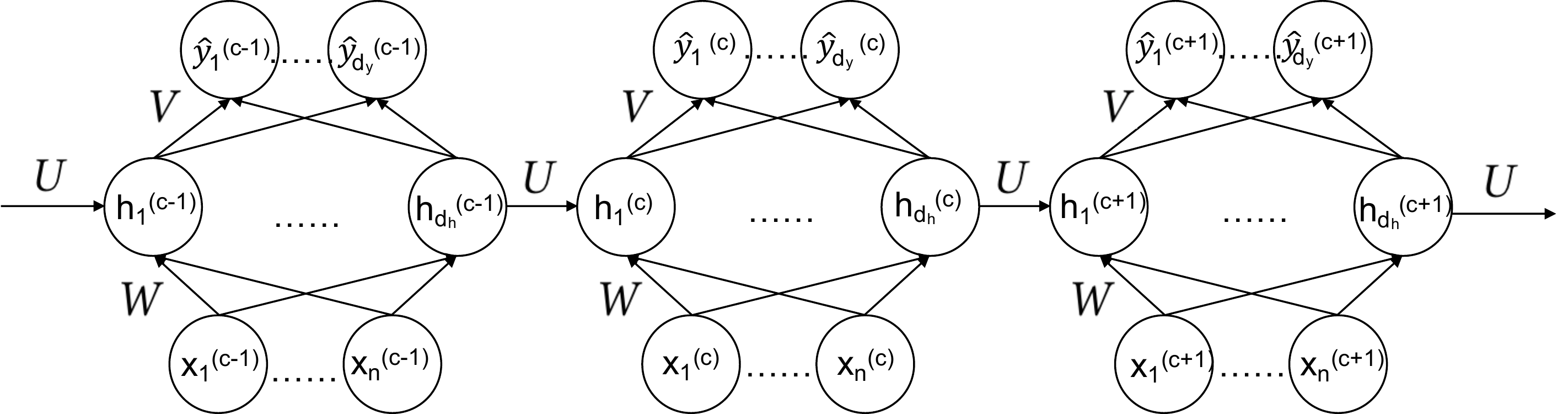}
	\caption{An example of recurrent neural network} 
	\label{fig:rnn}
\end{figure}

\paragraph{Original algorithm} Recall that the original non-distributed version of the RNN is divided into the forward propagation and backward propagation through time.
First, in the forward propagation, 
the output of the hidden layer propagated from the input layer is calculated as
\begin{equation}
\label{equ:forward1}
\begin{array}{ll}
& Z_h^{(c)} = X^{(c)}W + Uh^{(c-1)} + b_h\\
& h^{(c)} = \sigma_1(Z_h^{(c)})
\end{array}
\end{equation}
The output of the output layer propagated from hidden layer is calculated as
\begin{equation}
\label{equ:forward2}
\begin{array}{ll}
& Z_y^{(c)} = Vh^{(c)} + b_y   \\
& \hat{y}^{(c)} = \sigma_2(Z_y^{(c)})
\end{array}
\end{equation}
Then, the cost function $\sum\limits_{c=0}^{T} J(\hat{y}^{(c)}, y^{(c)})$, and the coefficient matrices $W,U,V$ are updated using the backward propagation through time.

The gradients of $J^{(c)}$ with respect to $V$ is calculated as $\frac{\partial{J^{(c)}}}{\partial{V}} = \frac{\partial{J^{(c)}}}{\partial{\hat{y}^{(c)}}}\frac{\partial{\hat{y}^{(c)}}}{\partial{Z_y^{(c)}}}\frac{\partial{Z_y^{(c)}}}{\partial{V}}$. We let $\frac{\partial{J^{(c)}}}{\partial{\hat{y}^{(c)}}} = \delta_{loss}^{(c)}$,  $\delta_{\hat{y}}^{(c)}= \sigma_2^\prime(Z_y^{(c)})$.
The gradient of $J^{(c)}$ with respect to $V$ can be written as:
\begin{equation}
\label{equ:back1}
\frac{\partial{J^{(c)}}}{\partial{V}} = [\delta_{loss}^{(c)}\circ\delta_{\hat{y}}^{(c)}](h^{(c)})^T.
\end{equation}

The gradients of $J^{(c)}$ with respect to $U$ is calculated as $\frac{\partial{J^{(c)}}}{\partial{U}} = \frac{\partial{J^{(c)}}}{\partial{\hat{y}^{(c)}}}\frac{\partial{\hat{y}^{(c)}}}{\partial{h^{(c)}}}\frac{\partial{h^{(c)}}}{\partial{U}}$, where $\frac{\partial{J^{(c)}}}{\partial{\hat{y}^{(c)}}}\frac{\partial{\hat{y}^{(c)}}}{\partial{h^{(c)}}} = V^T[\delta_{loss}^{(c)}\circ\delta_{\hat{y}}^{(c)}]$, and
\begin{align*}
 \frac{\partial{h^{(c)}}}{\partial{U}} = \sum\limits_{k=0}^{c}\frac{\partial{h^{(c)}}}{\partial{h^{(k)}}}\frac{\partial^+{h^{(k)}}}{\partial{U}}= \sum\limits_{k=0}^{c}\big(\prod\limits_{i=k+1}^{c}\sigma_1^\prime(Z_h^{(i)})U\big) \sigma_1^\prime(Z_h^{(k)})h^{(k-1)}.
\end{align*}
Let $\delta_{h}^{(c)} = \sigma_1^\prime(Z_h^{(c)})$, then the gradient of $J^{(c)}$ with respect to $U$ is
\begin{equation}
\label{equ:back2}
\frac{\partial{J^{(c)}}}{\partial{U}} = V^T\delta_{loss}^{(c)}\circ\delta_{\hat{y}}^{(c)}\bigg(\sum\limits_{k=0}^{c}\big(\prod\limits_{i=k+1}^{c}\delta_{h}^{(i)}U \big )\delta_{h}^{(k)}h^{(k-1)}\bigg).
\end{equation}

Similarly, the gradients of $J^{(c)}$ with respect to $W$ is calculated as:
\begin{equation}
\label{equ:back3}
\frac{\partial{J^{(c)}}}{\partial{W}} = V^T\delta_{loss}^{(c)}\circ\delta_{\hat{y}}^{(c)}\bigg(\sum\limits_{k=0}^{c}\big(\prod\limits_{i=k+1}^{c}\delta_{h}^{(i)}U \big )\delta_{h}^{(k)}X^{(k)}\bigg).
\end{equation}

\begin{algorithm}
	\caption{Privacy Preserving Collaborative Recurrent Neural Network}
	\label{alg:rnn}
	\begin{algorithmic}[1]
		\State{\textbf{Input:} Local training data $X^{l(c)}$ ($l=[1, ... , s] , c= [0,...,T]$), }
		\State{\qquad learning rate $\alpha$}
		\State{\textbf{Output:} Model parameters $W^l$ ($l=  [1, ... , s]$), $ U, V$}
		\State{\textbf{Initialize:} Randomize $W^l$ ($l=  [1, ... , s]$), $U, V$}
		\Repeat
		\State{\textbf{for all} $S_l \in S$ \textbf{do in parallel}}
		\State{\qquad $S_l: X^{l(c)}W^{l}, c= [0,...,T] $ }
		\State{\qquad $S_l: srt^{(c)} \leftarrow{X^{l(c)}W^{l}, c= [0,...,T]} $ }
		\State{\qquad $\{(S_l, E_{srt}^{(c)})\}_{S_l \in S} \leftarrow{Shr(srt^{(c)}, S)}$}
		\State{\textbf{end for}}
		\State{$X^{(c)}W \leftarrow{Rec(\{(S_l, E_{srt}^{(c)})\}_{S_l \in S})}$}
		\State{$A: Z_h^{(c)} \leftarrow {X^{(c)}W + Uh^{(c-1)} + b_h}$,  $h^{(c)} \leftarrow \sigma_1(Z_h^{(c)})$}
		\State{$A: Z_y^{(c)} \leftarrow Vh^{(c)} + b_y$, $\hat{y}^{(c)} = \sigma_2(Z_y^{(c)})$}
		\State{$A: \delta_{loss}^{(c)} \leftarrow \frac{\partial{J^{(c)}}}{\partial{\hat{y}^{(c)}}}, \delta_{\hat{y}}^{(c)} \leftarrow \sigma_2^\prime(Z_y^{(c)}), \delta_{h}^{(c)} = \sigma_1^\prime(Z_h^{(c)}) $}
		\State{$A: \frac{\partial{J^{(c)}}}{\partial{V}} \leftarrow [\delta_{loss}^{(c)}\circ\delta_{\hat{y}}^{(c)}](h^{(c)})^T$}
		\State{$A: \frac{\partial{J^{(c)}}}{\partial{U}} \leftarrow V^T\delta_{loss}^{(c)}\circ\delta_{\hat{y}}^{(c)}\bigg(\sum\limits_{k=0}^{c}\big(\prod\limits_{i=k+1}^{c}\delta_{h}^{(i)}U \big )\delta_{h}^{(k)}h^{(k-1)}\bigg)
			$}
		\State{$A: V \leftarrow V - \alpha \sum\limits_{c=0}^{T} \frac{\partial{J^{(c)}}}{\partial{V}}$}
		\State{$A: U \leftarrow U - \alpha \sum\limits_{c=0}^{T} \frac{\partial{J^{(c)}}}{\partial{U}}$}		
		\State{\textbf{for all} $S_l \in S$ \textbf{do in parallel}}
		\State{\qquad  $\frac{\partial{J^{(c)}}}{\partial{W^l}} = V^T\delta_{loss}^{(c)}\circ\delta_{\hat{y}}^{(c)}\bigg(\sum\limits_{k=0}^{c}\big(\prod\limits_{i=k+1}^{c}\delta_{h}^{(i)}U \big )\delta_{h}^{(k)}X^{l(k)}\bigg)
			$ }
		\State{\qquad $W^l \leftarrow W^l - \alpha \sum\limits_{c=0}^{T} \frac{\partial{J^{(c)}}}{\partial{W^l}}$}
		\State{\textbf{end for}}
		\Until{Convergence} 
	\end{algorithmic}
\end{algorithm}

\paragraph{\codename algorithm} In \codename, each local node keeps $X^{l(c)}, c= [0,...,T]$, and maintains the coefficient matrix $W^l$ locally. The aggregation node maintains coefficient matrices $U, V$. Similar to its original non-distributed counterpart, the training process is divided into the forward propagation and backward propagation through time. Below we briefly outline these steps and the detailed algorithm is given by Algorithm \ref{alg:rnn}.



In the forward propagation, local nodes compute $X^{l(c)}W^{l}, c= [0,...,T]$ respectively (line 7)  (line number in Algorithm \ref{alg:rnn}), and $X^{(c)}W = \sum\limits_{l=1}^{s} X^{l(c)}W^l$ is calculated using the secret-sharing scheme (line 8 to line 11). The aggregation node then computes $Z_h^{(c)}, h^{(c)},  Z_y^{(c)}, \hat{y}^{(c)}$ using Equation \ref{equ:forward1} and \ref{equ:forward2} (line 12 and line 13).



In the backward propagation through time, for each $J^{(c)}$ at time $t$, the aggregation node computes $\delta_{loss}^{(c)}, \delta_{\hat{y}}^{(c)}, \delta_{h}^{(k)}, k=[0,...,t] $, and sends $\delta_{loss}^{(c)}, \delta_{\hat{y}}^{(c)}, \delta_{h}^{(k)}$ to local nodes (line 14). Then the aggregation node computes the gradients of $J^{(c)}$ with respect to $V, U$ using Equation \ref{equ:back1} and \ref{equ:back2} (line 15 and line 16), and updates $U, V$ (line 17 and line 18). The local nodes compute the gradients of $J^{(c)}$ with respect to $W^{l}$ using Equation \ref{equ:back4}, and update $W^l$ respectively (line 19 to line 22).



\begin{equation}
\label{equ:back4}
\frac{\partial{J^{(c)}}}{\partial{W^l}} = V^T\delta_{loss}^{(c)}\circ\delta_{\hat{y}}^{(c)}\bigg(\sum\limits_{k=0}^{c}\big(\prod\limits_{i=k+1}^{c}\delta_{h}^{(i)}U \big )\delta_{h}^{(k)}X^{l(k)}\bigg).
\end{equation}



\section{Performance Evaluation}
\label{sec:performance}

We implement \codename in C++. It uses the Eigen library \cite{eigenweb} to handle matrix operations, and uses ZeroMQ library \cite{hintjens2013zeromq} to implement the distributed messaging.
The experiments are executed on four Amazon EC2 c4.8xlarge machines 
with 60GB of RAM each, three of which act as local nodes and the other acts as the aggregation node.
To simulate the real-world scenarios, we execute \codename on both LAN and WAN network settings. In the LAN setting, machines are hosted in a same region, and the average network bandwidth is 1GB/s. 
In the WAN setting, we host these machines in different continents. 
The average network latency (one-way) is 137.7ms, and the average network throughput is 9.27MB/s.
We collect 10 runs for each data point in the results and report the average.
We use the MNIST dataset \cite{lecun-mnisthandwrittendigit-2010},
and duplicate its samples when its size is less than the sample size $m$~($m \geq 60,000$).

We take non-private machine learning which trains on the concatenated dataset as the baseline, and  compare with MZ17 \cite{mohassel2017secureml}, which is the state-of-the-art cryptographic solution for privacy preserving machine learning.
It is based on oblivious transfer (MZ17-OT) and linearly homomorphic encryption (MZ17-LHE). 
As shown in Fig. \ref{fig:comparison}, \codename achieves significant efficiency improvement over MZ17, and due to parallelization in the computing of the local nodes, 
\codename also outperforms the non-private baselines in the LAN network setting. 


\textbf{Linear regression and logistic regression}. We use mini-batch stochastic gradient descent (SGD) for training the linear regression and logistic regression. We set the batch size $|B| = 40$ with 4 sample sizes ($1,000$-$100,000$) in the linear regression and logistic regression. 

In the LAN setting, \codename achieves around $45$x faster than MZ17-OT. It takes $13.11s$ for linear regression (Figure \ref{fig:lrlan}) and $12.73s$ for logistic regression (Figure \ref{fig:loglan}) with sample size $m = 100,000$, while in MZ17-OT, $594.95s$ and $605.95s$ are reported respectively. \codename is also faster than the baseline, which takes $17.20s$  and  $17.42s$  for linear/logistic regression respectively.
In the WAN setting, \codename is around $9$x faster than MZ17-LHE.
It takes $1408.75s$ for linear regression (Figure \ref{fig:lrwan}) and $1424.94s$ for logistic regression (Figure \ref{fig:logwan}) with sample size $m = 100,000$, while in MZ17-LHE, it takes $12841.2s$ and $13441.2s$ respectively with the same sample size. 
It is worth mentioning that in MZ17, an MPC-friendly alternative function is specifically designed to replace non-linear sigmoid functions for training logistic regression, while in our framework, the non-linear function is used as usual.
To further break down the overhead to computation and communication, we
summarize the results of linear regression and logistic regression on other sample sizes in Table \ref{tab:lrlog}.

\begin{figure}[H]%
	\centering
	\subfloat[Linear Regression LAN]{{\includegraphics[width=4.0cm]{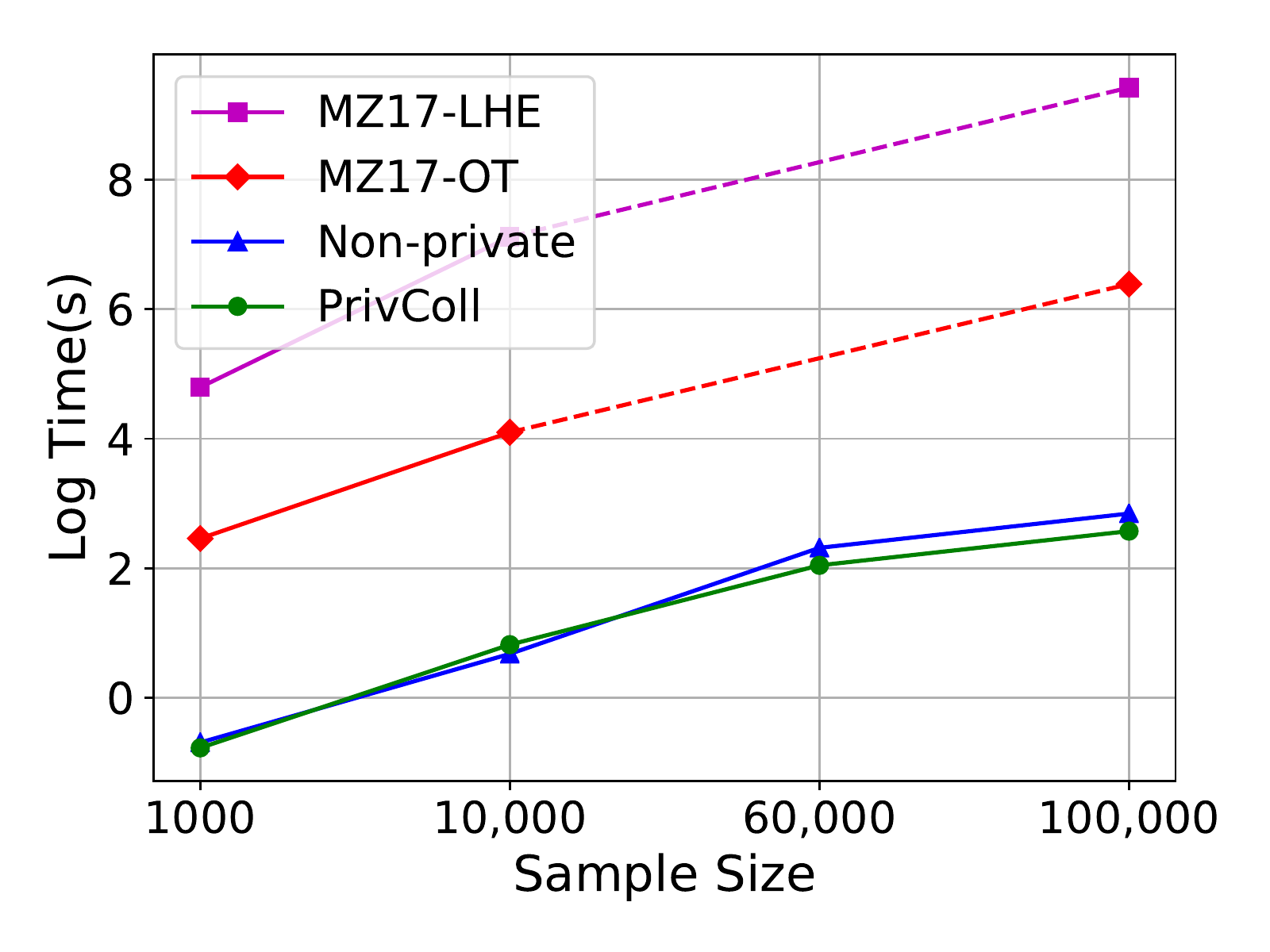} \label{fig:lrlan}}}%
	\subfloat[Logistic Regression LAN]{{\includegraphics[width=4.0cm]{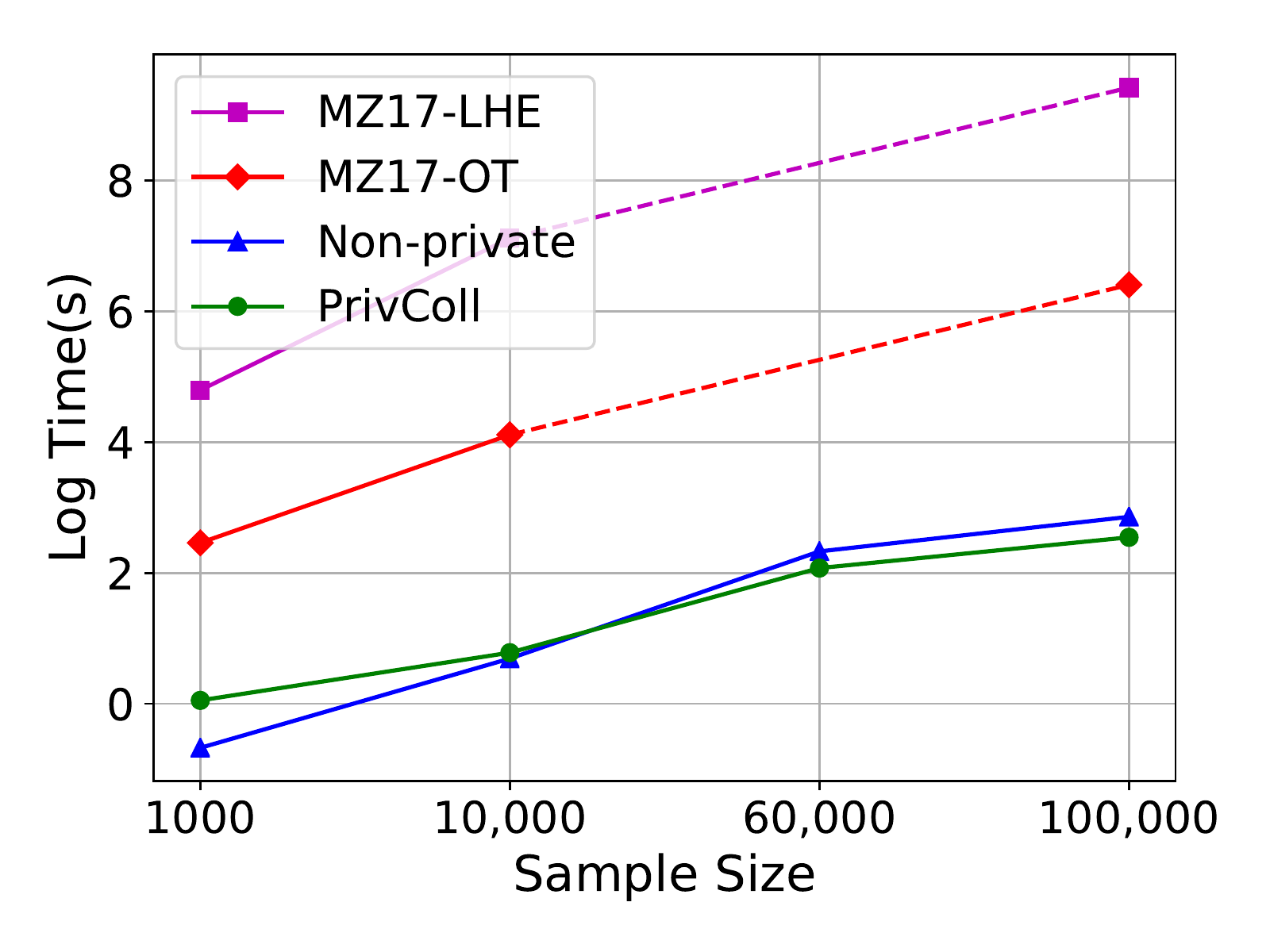} \label{fig:loglan}}}%
	\subfloat[Neural Network LAN]{{\includegraphics[width=4.0cm]{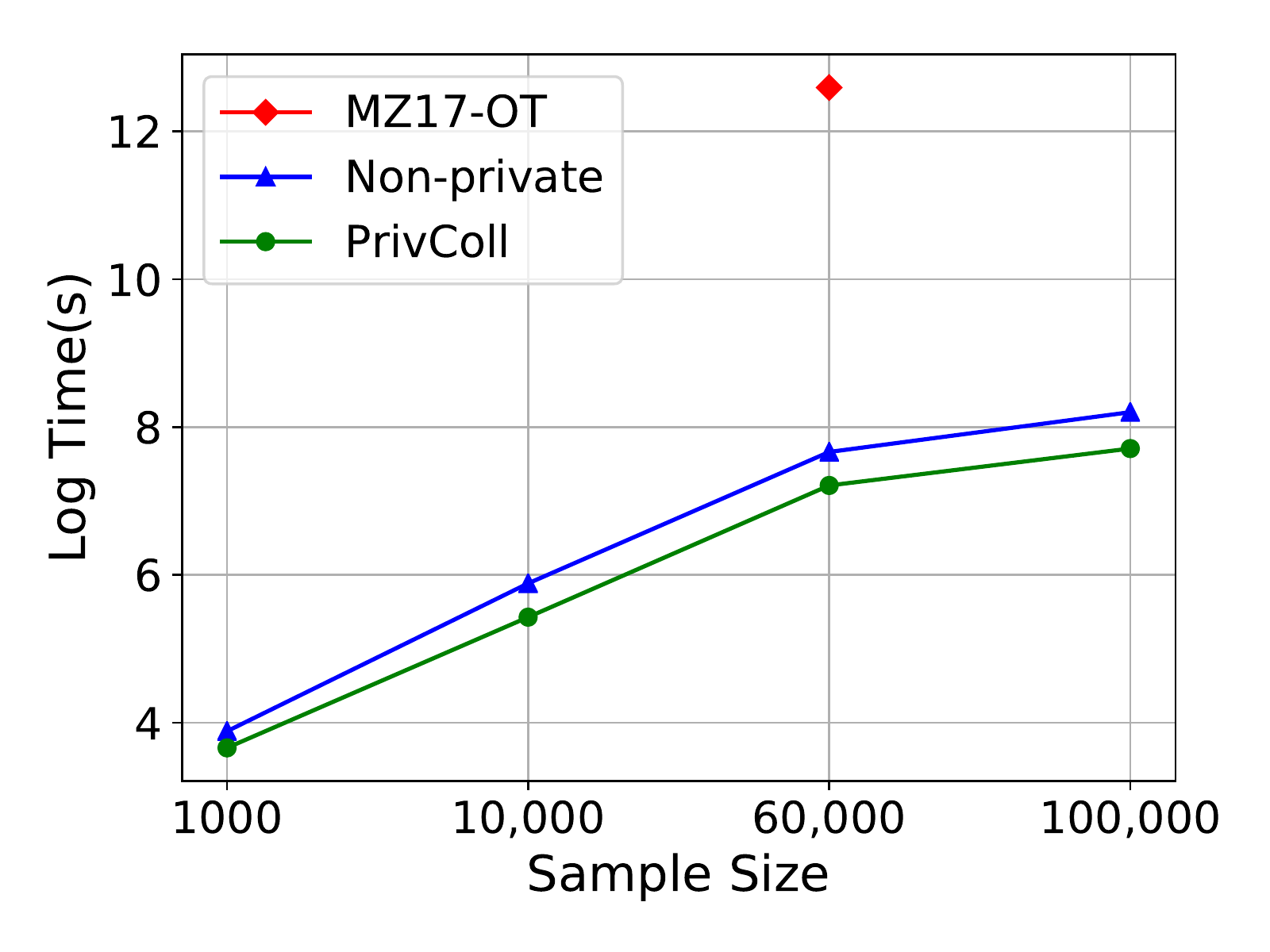} \label{fig:nnlan}}}%
	\qquad
	\subfloat[Linear Regression WAN]{{\includegraphics[width=4.0cm]{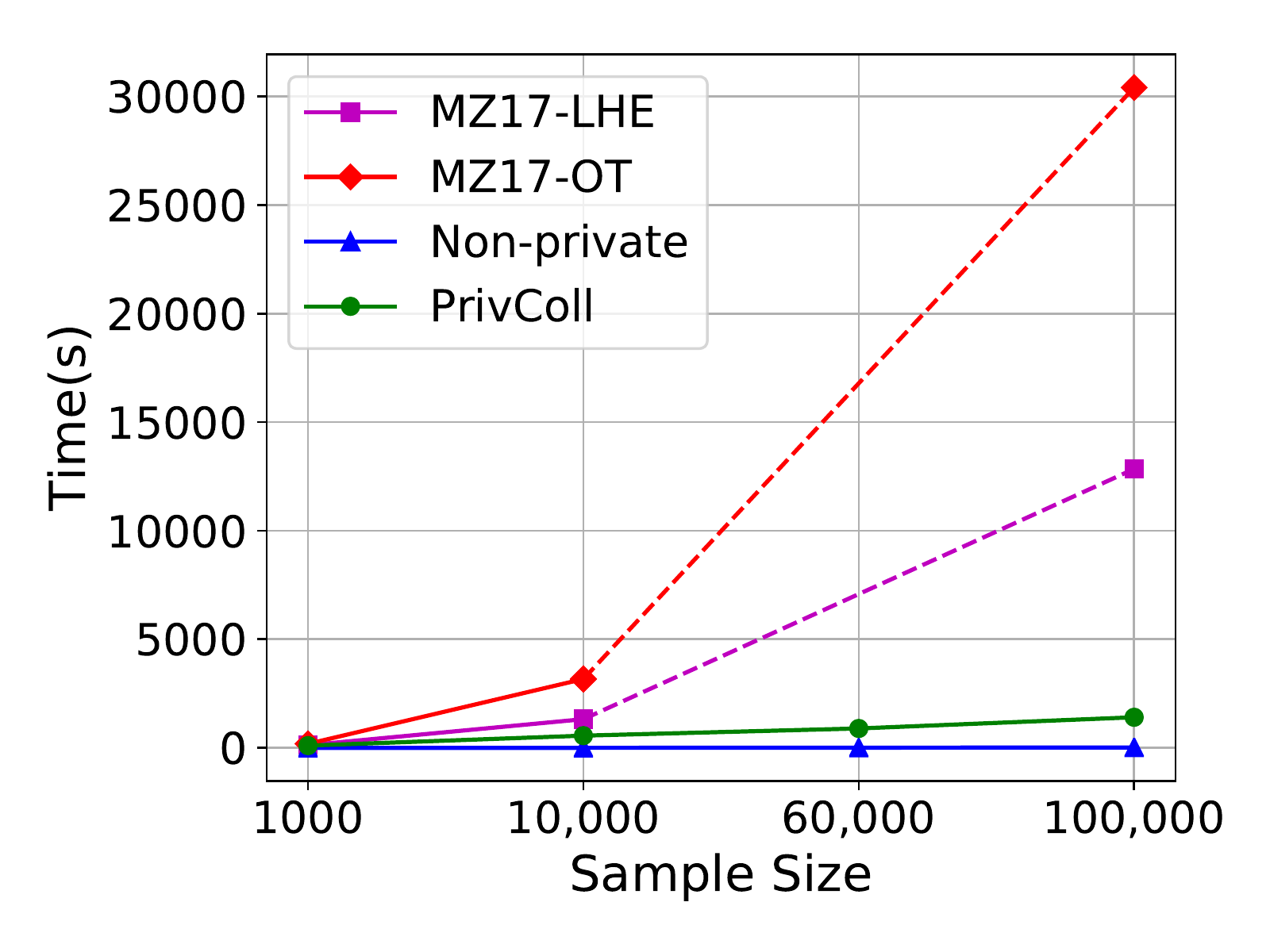} \label{fig:lrwan}}}%
	\subfloat[Logistic Regression WAN]{{\includegraphics[width=4.0cm]{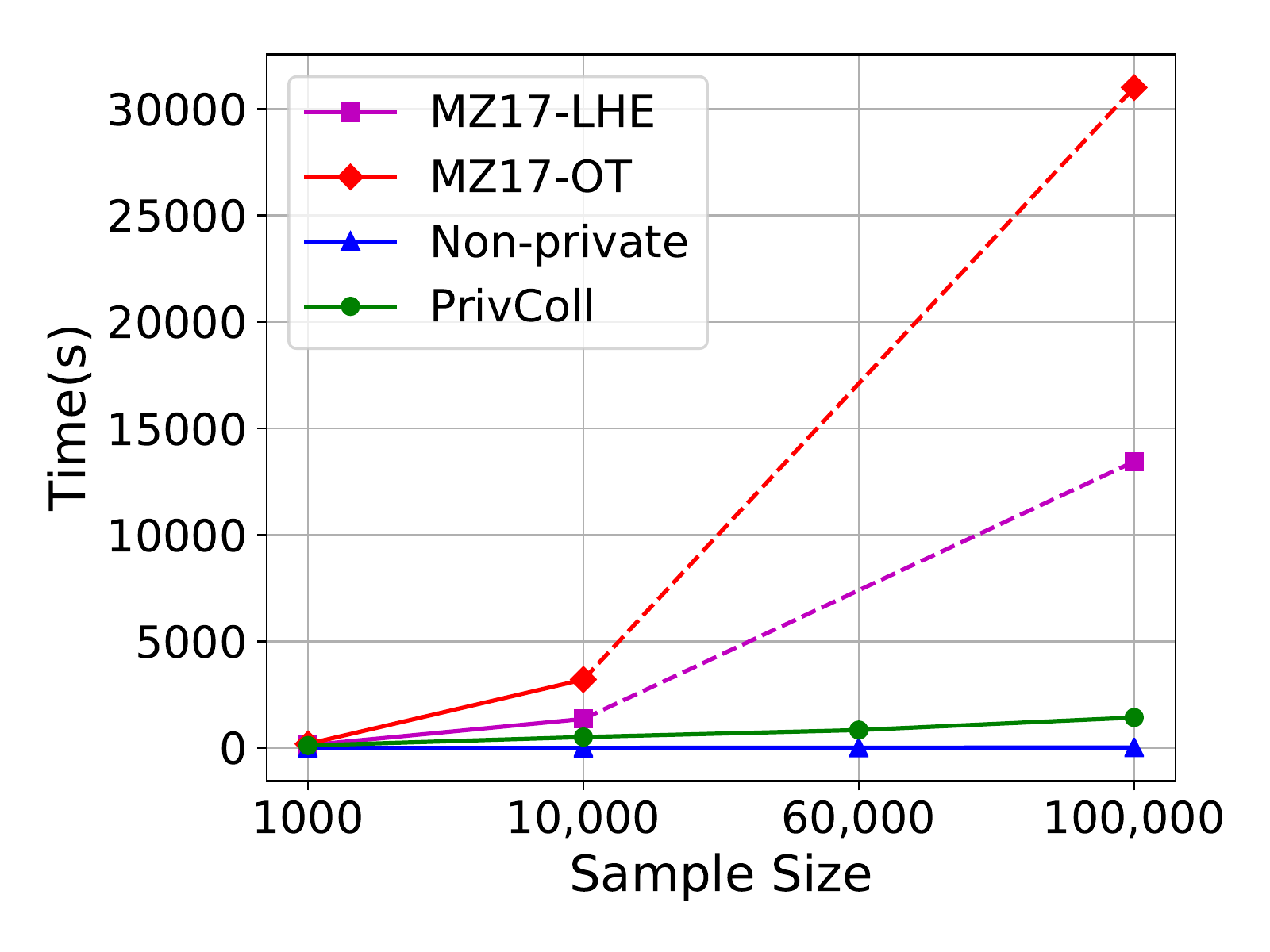} \label{fig:logwan}}}%
	\subfloat[Neural Network WAN]{{\includegraphics[width=4.0cm]{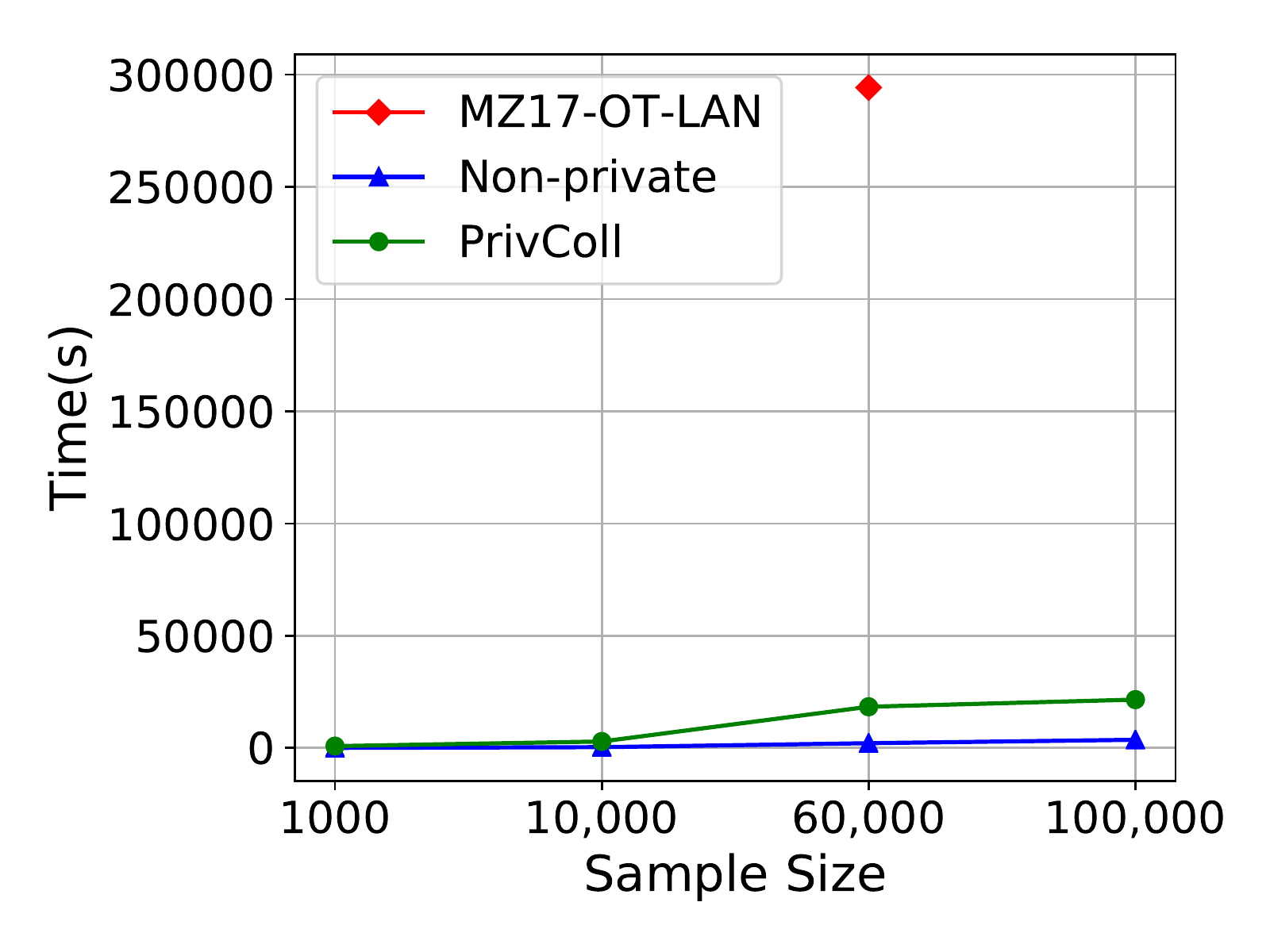} \label{fig:nnwan}}}%
	%
	\caption{Efficiency comparison. a-c) the natural logarithm of running time(s) as the sample size increases. d-f) running time(s) as the sample size increases. 
	}%
	\label{fig:comparison}
\end{figure}

\begin{table*}[]
	\caption{\codename's Overhead Breakdown for Linear/Logistic Regression}
	\label{tab:lrlog}
	\resizebox{\textwidth}{!}{
		\begin{tabular}{|c|c|c|c|c|c|c|c|c|c|c|}
			\hline
			\multirow{3}{*}{} & \multicolumn{5}{c|}{Linear Regression}                                                         & \multicolumn{5}{c|}{Logistic Regression}                                                       \\ \cline{2-11}
			& \multirow{2}{*}{Computation} & \multicolumn{2}{c|}{Communication} & \multicolumn{2}{c|}{Total} & \multirow{2}{*}{Computation} & \multicolumn{2}{c|}{Communication} & \multicolumn{2}{c|}{Total} \\ \cline{3-6} \cline{8-11}
			&                              & LAN              & WAN             & LAN         & WAN          &                              & LAN             & WAN              & LAN         & WAN          \\ \hline
			m=1,000           & 0.445s                       & 0.016s          & 103.28s          & 0.461s      & 103.73s       & 0.467s                       & 0.583s          & 109.57s           & 1.050s      & 110.04s       \\ \hline
			m=10,000          & 1.293s                       & 0.978s           & 561.53s         & 2.271s       & 562.83s      & 1.368s                       & 0.812s          & 506.64s          & 2.180s      & 508.01s      \\ \hline
			m=60,000          & 6.155s                        & 1.578s           & 887.63s         & 7.733s       & 893.79s      & 6.271s                        & 1.683s          & 829.69s          & 7.954s       & 835.96s      \\ \hline
			m=100,000         & 10.07s                        & 3.037s           & 1398.67s         & 13.11s      & 1408.75s      & 10.29s                        & 2.434s          & 1414.64s          & 12.73s      & 1424.94s       \\ \hline
		\end{tabular}
	}
\end{table*}

\textbf{Neural Network}.
We implement a fully connected neural network in \codename. It has two hidden layers with 128 neurons in each layer (same as MZ17) and takes a sigmoid function as the activation function.
For training the neural network, we set the batch size $|B| = 150$ with 4 sample sizes ($1,000$-$60,000$). 

In the LAN network setting, \codename achieves  $1352.87s$ (around $22.5$ minutes) (Figure \ref{fig:nnlan}) with sample size $m = 60,000$, while in MZ17, it takes  $294,239.7s$ (more than $81$ hours) with the same sample size.
\codename also outperforms the non-private baseline which takes $2,127.87s$.
In the WAN setting,  \codename achieves $18,367.88s$ (around $5.1$ hours) with sample size $m = 60,000$ (Figure \ref{fig:nnwan}), while in MZ17, it is not yet practical for training neural networks in WAN setting due to the high number of interactions and high communication. Note that, in Figure \ref{fig:nnwan}, we still plot the MZ17-OT-LAN result ($294,239.7s$), showing that even when running our framework in the WAN setting, it is still much more efficient compared to the MPC solutions in MZ17 run in the LAN setting.
The overhead breakdown on computation and communication is summarized in Table \ref{tab:nn}.

\begin{table}[]
	\centering
	\caption{\codename's Overhead Breakdown for Neural Network}
	\label{tab:nn}
	\resizebox{8.7cm}{!}{
		\begin{tabular}{|c|c|l|c|c|c|c|}
			\hline
			\multirow{3}{*}{} & \multicolumn{6}{c|}{Neural Network}                                                                                 \\ \cline{2-7}
			& \multicolumn{2}{c|}{\multirow{2}{*}{Computation}} & \multicolumn{2}{c|}{Communication} & \multicolumn{2}{c|}{Total} \\ \cline{4-7}
			& \multicolumn{2}{c|}{}                             & LAN            & WAN               & LAN         & WAN          \\ \hline
			m=1,000           & \multicolumn{2}{c|}{30.14s}                       & 8.729s          & 795.27s           & 38.87s      & 825.42s      \\ \hline
			m=10,000          & \multicolumn{2}{c|}{223.08s}                      & 4.683s          & 2662.58s          & 227.76s     & 2885.66s     \\ \hline
			m=60,000          & \multicolumn{2}{c|}{1320.77s}                     & 32.10s         & 17047.10s          & 1352.87s    & 18367.88s     \\ \hline
			m=100,000         & \multicolumn{2}{c|}{2180.74s}                     & 47.99s         & 19364.73s         & 2228.73s    & 21545.47s    \\ \hline
		\end{tabular}
	}
\end{table}



 \section{Related Work}
\label{sec:related}

The studies most related to \codename are~\cite{zheng2019towards,hu2019fdml}. 
Zheng et al~\cite{zheng2019towards} employ the lightweight additive secret sharing scheme for secure outsourcing of the decision tree algorithm for classification. 
Hu et al~\cite{hu2019fdml} propose FDML, which is a collaborative machine learning framework for distributed features, and the model parameters are protected by additive noise mechanism within the framework of differential privacy. 

There also have been some previous research efforts which have explored collaborative learning without exposing their trained models~\cite{jia2016privacy,zhang2019enabling, papernot2018scalable,zhang2020differentially}. For example, Papernot et al~\cite{papernot2018scalable} make use of transfer learning in combination with differential privacy to learn an ensemble of \emph{teacher} models on data partitions, and then use these models to train a private \emph{student} model.

In addition, there are more works on generic privacy-preserving machine learning frameworks  via HE/MPC solutions \cite{kwabena2019mscryptonet, zhang2017private, esposito2018securing, yuan2014privacy, hardy2017private, sadat2017safety, gilad2016cryptonets, will2015secure, ko2015stratus,ryffel2018generic} or differential privacy mechanism~\cite{abadi2016deep, gupta2018distributed,vepakomma2018split,abuadbba2020can}. Recent studies \cite{gascon2017privacy, gascon2016secure}  propose a hybrid multi-party computation protocol for securely computing a linear regression model.
In \cite{liu2017oblivious}, an approach is proposed for transforming an existing neural network to an oblivious neural network supporting privacy-preserving predictions.
In \cite{yuan2014privacy}, a secure protocol is presented to calculate the delta function in the back-propagation training.
In \cite{mohassel2017secureml}, a MPC-friendly alternative function is specifically designed to replace non-linear sigmoid and softmax functions, as the division and the exponentiation in these function are expensive to compute on shared values.

\section{Conclusion}
\label{sec:conclusion}
We have presented \codename, a practical privacy-preserving collaborative machine learning framework.  
\codename guarantees privacy preservation for  both local training  data  and   models trained on them, 
against an honest-but-curious adversary. 
It also ensures the correctness of a wide range of machine/deep learning algorithms, such as  linear regression, logistic regression, and a variety of neural networks. 
Meanwhile, \codename achieves a practical applicability. 
It is much more efficient compared to other state-of-art solutions.

%
%
%

%
%
%
 \bibliographystyle{splncs04}
 \bibliography{esorics20bib}

\end{document}